\documentclass{llncs}
\usepackage{amssymb}
\usepackage{enumerate}
\usepackage[english]{babel}
\usepackage{amsmath}
\usepackage{algorithmic}

\usepackage[T1]{fontenc}

\usepackage{comment}
\usepackage{graphicx,amsmath,amssymb}
\spnewtheorem{claimu}[theorem]{Claim}{\bfseries}{\itshape}

\begin{document}
\title{Parameterized Complexity of Edge Interdiction Problems}

\author{
Jiong Guo\thanks{Supported by the DFG Excellence Cluster MMCI.}
\and
Yash Raj Shrestha\thanks{Supported by the DFG research project DARE GU~1023/1.}
}
 
\institute{
Universit\"at des Saarlandes, \\
Campus E 1.7, 
D-66123 Saarbr\"ucken, Germany\\
\email{\{jguo, yashraj\}@mmci.uni-saarland.de}
}

\maketitle
\begin{abstract}

We study the parameterized complexity of graph interdiction problems. For an optimization problem on graphs, one can formulate an interdiction
problem as a game consisting of two players, namely, an interdictor and an evader, who compete on an objective with opposing interests. In edge interdiction problems, every
edge of the input graph has an interdiction cost associated with it and the interdictor interdicts the graph by modifying the edges in the graph, and the number of such modifications
is constrained by the interdictor's budget. The evader then solves the given optimization problem on the modified graph. The action of the interdictor must impede the evader as much as
possible. 
 
~~~~~We focus on edge interdiction problems related to minimum spanning tree, maximum matching and maximum flow problems. These problems arise in different real world scenarios. 
We derive several fixed-parameter tractability and W[1]-hardness results for these interdiction
problems with respect to various parameters. Hereby, we show close relation between edge interdiction problems and partial covering problems on bipartite graphs. 
\end{abstract}

 \section{Introduction}
 Given an optimization problem on graphs, the corresponding interdiction problem can be formulated as a game consisting of two players, 
 namely, an interdictor and an evader, who compete on an objective with opposite interests. In edge interdiction problems, 
 every edge of the input graph has an interdiction cost associated with it and the interdictor interdicts the network by modifying edges in the graph, and the 
 number of such modifications are constrained by the interdictor's budget. The evader then solves the given optimization problem on the modified graph. 
 The action of the interdictor must impede the evader as
 much as possible. 

 In this paper, we focus on edge interdiction problems related to minimum spanning tree, maximum matching and maximum flow problems. 
 These interdiction problems arise in different real world scenarios, e.g., detecting drug smuggling \cite{Wood93,Wood95}, military planning \cite{ghare71}, analyzing power grid
 vulnerability \cite{salmeron_09} and hospital infection control \cite{n87}.

 A {\it spanning tree} of a connected graph $G$ is a tree composed of all the vertices and some of the edges of $G$. The {\it minimum spanning tree (MST) problem} is to find
 a spanning tree whose total weight is minimum. Let $\eta(G)$ be the weight of MST of G. 
 A {\it matching} in a graph is a set of edges such that no two edges share an endpoint. Let $\nu(G)$ be the weight of maximum matching in $G$.
 We say a matching $M$ {\it saturate} a set $U \subseteq V$ if for each vertex $u \in U$, there exists one edge in $M$ with $u$ as its endpoint.
 For $G=(V,E)$, $G -I$ is the graph resulting by
  removing a set of edges $I$ from $G$. A set of edges $M$ of $G=(V,E)$ is called an {\it edge dominating set} if every edge of $E \setminus M$ is adjacent to at least one edge in $M$.
 An {\it independent edge dominating set} is an edge dominating set in which no two edges are adjacent.
 A {\it minimum maximal matching} in a graph $G$ is a maximal matching of
 the minimum size, denoted by $\lambda (G)$. An independent edge dominating set is a minimum maximal matching \cite{forcade}.

We start with the introduction of $b$-{\sc Most Vital Edges in MST ($b$-MVE)} which is defined in the literature follows: 

 \begin{description}\itemsep-1pt
  \item {\bfseries Input:}  An edge-weighted graph $G=(V,E)$ with weight function $w: E \rightarrow \mathbb{Z}_{\geq 0}$, two positive integers $b$ and $r$.
  \item {\bfseries Output:}  A subset $I \subseteq E$ with $|I| \leq b$ such that  $\eta(G - I) \geq r$.  
 \end{description}

 Frederickson and Solis-Oba \cite{FredericksonS99} proved that $b$-{\sc MVE} is NP-hard even 
 if the weights of the edges are either 0 or 1. They also gave an $ \Omega(1\setminus \log k)$-approximation algorithm for $b$-{\sc MVE}. This problem has also been studied from view point of exact
algorithms and randomized algorithms \cite {Liang01,LiangS97}. 
 
 The {\sc Maximum Matching Edge Interdiction} (MMEI) problem, introduced by Zenklusen \cite{Zenklusen10a}, is defined as follows: 
  
 \begin{description}\itemsep-1pt
  \item {\bfseries Input:} A weighted graph $G=(V,E)$ with weight function $w: E \rightarrow \mathbb{Z}_{\geq 0}$, an interdiction cost function
                 $c : E \rightarrow \mathbb{Z}_{\geq 1}$, and two positive integers $b$ and $m$.  
  \item {\bfseries Output:} A subset $I \subseteq E$ with $c(I) \leq b$ such that $\nu(G - I) \leq m$.  
 \end{description}
MMEI is NP-hard on bipartite graphs, even with unit edge weights and unit interdiction costs \cite{Zenklusen10a}. Zenklusen \cite{Zenklusen10a} introduced a constant
factor approximation algorithm for MMEI on graphs with unit edge weights. Recently, Dinitz and Gupta provided a constant-factor approximation for a generalization of matching
 interdiction called {\it packing interdiction}
\cite{gupta}. Zenklusen \cite{Zenklusen10a} also showed that MMEI is solvable in pseudo-polynomial time on graphs with bounded treewidth. Pan and Schild \cite{pan} proved 
that weighted MMEI remains NP-hard even
on planar graphs and gave a pseudo-polynomial time approximation scheme for the same on planar graphs. 

The $s$-$t$ {\sc Flow Edge Interdiction} ($s$-$t$ FEI) is defined as follows~\cite{ghare71}: 

 \begin{description}\itemsep-1pt
  \item {\bfseries Input:} A directed graph $G=(N,A)$ with distinguished vertices $s$ and $t$,  positive integer capacity $u_{ij}$ for each arc $(i,j) \in A$, an interdiction cost function~$c : A \rightarrow \mathbb{Z}_{\geq 0}$,  
                   two positive integers $b$ and $r$. 
  \item {\bfseries Output:} A set of arcs $A'$ with $c(A') \leq b$ such that the maximum $s$-$t$ flow in~$G-A'$ has value at most $r$. 
 \end{description}

Philips \cite{philip} proved that $s$-$t$ FEI is strongly NP-complete for general graphs, even of degree at most 3. 
Pan and Schild \cite{pan} recently studied some variants of $s$-$t$ FEI on planar graphs.

 We study the parameterized complrtial coverexity \cite{Nie06} of the three edge interdiction problems defined above. First, we show that $b$-MVE is W[1]-hard
 with respect to $r$. In graphs with edges of weights only 0 or 1, $b$-MVE is FPT with respect to $b$. The reduction from {\sc Clique} to MMEI by Zenklusen \cite{abs-0804-3583} already
shows that MMEI with respect to $b$ is W[1]-hard, even in bipartite input graphs with  unit edge weights and interdiction costs . Complementing this result, we prove that MMEI
with $m$  as parameter is W[1]-hard as well, even in graphs with unit edge weights and interdiction cost. In contrast to parameter $b$, MMEI becomes fixed-parameter tractable (FPT) with respect to $m$, 
if further restricted to bipartite input graphs. Moreover, taking both $b$ and $m$ as parameters leads also to FPT when restricted to instances with unit edge weights. Concerning $s$-$t$ FEI, we prove that the parameterization
with $b$ is W[1]-hard, complementing the result by Wood \cite{Wood93}, that $s$-$t$ FEI is W[1]-hard with respect to $r$.

 We observe some close relation between partial covering problems on bipartite graphs and edge interdiction problems. The goal of a partial covering problem 
 is not to cover all elements but to minimize/maximize the number
 of covered elements with a specific number of sets. For instance, the {\sc Partial Vertex Cover} ($k$-PVC) problem asks for $k$ vertices
maximizing the number of covered edges. Partial covering problems have been studied intensively not only because they generalize classical covering problems,
 but also because of many real life applications, see for example \cite{AroraK06,Bar-Yehuda01,AminiFS11,FominLRS11}.  

Our findings about the relation between partial covering problems and edge interdiction problems can be summarized as follows: First, we give a parameterized reduction from the W[1]-hard
$k$-PVC problem to MMEI, leading to the W[1]-hardness of MMEI with respect to $m$. Then, we prove an equivalent relation between a special version of MMEI and $k$-PVC on bipartite
graphs and thus derive the FPT result of this special case of MMEI. Moverover, we prove the W[1]-hardness of $k$-PVC on bipartite graphs with the number of uncovered edges as parameter by
a reduction from MMEI with respect to parameter $b$. Further, we introduce a new edge interdiction problem which turns out to be equivalent to the {\sc Partial Edge Dominating Set}
problem and prove W[1]-hardness for both. Finally, we study another bipartite version of $k$-PVC.

\paragraph*{Preliminary}
  
  For a vertex $v$, the vertices which are adjacent to $v$ in $G$ form the neighborhood $N(v)$ of $v$. For $U \subseteq V$, let $N(U)$ denote the set of all vertices which are adjacent to those in $U$. 
  We denote the size of the neighborhood of $v$ in $G$ as $\text{deg}_G(v)$. A {\it degree-1}
  vertex is a vertex with $\text{deg}_G(v)=1$. A {\it path} from vertex $a$ to vertex $b$ is an ordered sequence $a=v_0,~v_1, \dots,~v_m=b$ of distinct
  vertices in which each adjacent pair $(v_{j-1}, v_j)$ is linked by an edge. The {\it distance} between two vertices is the number of edges on
the shortest path between them, while the distance between two edges $e_1$ and $e_2$ is the minimum of the distances of their endpoints.
  A subgraph $H=(V',E')$ of a graph $G=(V,E)$
  is a pair $V' \subseteq V$ and $E' \subseteq E$. We say that $H=(V',E')$ is an {\it induced} subgraph of $G$ if $V' \subseteq V$ and $E'=\{\{u,v\} \in E | u, v \in V'\}$ 
  and we denote $H=G[V']$. For a set of edges $S$, let $V(S)$ denote the set of endpoints of $S$. An edge $e$ is {\it dominated} by another edge $e'$
  if they share at least one endpoint. An edge $e$ is {\it } covered by a vertex $v$ if $v$ is one of the endpoints of $e$.  A {\it disconnecting set} of edges $F$ is such that $G-F$ has more connected components than $G$ does. A connected graph $G$ is $k$-{\it edge-connected}
  if every disconnecting edge set has at least $k$ edges.

 \section{\sc $b$-Most Vital Edges in MST}
  We consider two parameterizations for $b$-MVE. Firstly, we define {\sc Minimum $k$-way Edge Cut} which is vital in the proofs of the following two results.

  \begin{description}\itemsep-1pt
  \item {\sc \bfseries Minimum $k$-way Edge cut}  
  \item {\bfseries Input:} An undirected graph~$G=(V,E)$ with unit edge weight and non-negative integers~$k$ and $s$.  
  \item {\bfseries Question:} Is there a set $S \subseteq E$ with $|S| \leq s$ such that, $G-S$ has at least $k$ connected components?   
 \end{description}
 
  \begin{theorem} \label{thm:mstwhard}
  $b$-{\sc Most Vital Edges} is W[1]-hard with respect to the weight $r$ of the MST in $G-I$.
  \end{theorem}

 \begin{proof}
  We give a parameterized reduction from {\sc Minimum $k$-Way Edge Cut}.
 Downey et al. \cite{DowneyEFPR03} proved that {\sc Minimum $k$-Way Edge Cut} is W[1]-hard with parameter $k$. 
 Given an instance $G=(V,E)$ for {\sc Minimum $k$-Way Edge Cut}, we create
 an instance $G'=(V',E')$ for $b$-{MVE} as follows: For each vertex $v \in V$, we create a vertex $v \in V'$.
 For each edge $\{u,v\} \in E$, create an edge 
 $\{u,v\} \in E'$ with weight 0.
 For each pair of vertices $u,v$ in $V'$ such that $\{u,v\} \notin E$, we add a connection gadget $M$ between 
 $u$ and $v$ in the following way: 
 Create a clique with $b+1$ vertices as gadget $M$
 such that all edges in $M$ have weight 0. Now, connect $u$ to all vertices in 
 $M$ with edges of weight 1 and $v$ to all 
 vertices in $M$ with edges of weight 0. Let $X'$ be the set of all vertices in connection gadgets and $Y'$ be the set of edges in $G'$ with at least one endpoint in $X'$. 
 Now we show that $G$ has a $k$-way edge cut of size $s$ iff at most $b=s$ edges can be deleted from $G'$ such that MST of the remaining graph is at least $r=k-1$. 

 ($\Rightarrow$)
 Given the graph $G$, let $S$ be the solution for {\sc Minimum $k$-way Edge Cut}. Now, $G-S$ consists of at least $k$ connected components. 
 We take the edges in $G'$ corresponding to those in 
 $S$ as the solution $S'$ for {$b$-MVE}. We can observe that $G'[V' \setminus X']-S'$ consists of at least $k$ connected components. 
 Hence, every spanning tree of $G'-S'$ must pass through at least $k-1$ 
 connection gadgets $M$. Let $M'$ be one connection gadget through which the MST in $G'-S'$ passes. 
 Let $M'$ be connected to $u \in V' \setminus X'$ with edges of weight 0 and to $v \in V' \setminus X'$
 with edges of weight 1. The MST can span all vertices in $M'$ by using $b+1$ edges between $u$ and $M'$ and connect $u$ to $v$ by taking 
 exactly one edge between $M'$ and $v$. The MST must pass through at least
 $k-1$ such connection gadgets. Moreover, each connected component of $G'[V' \setminus X']-S'$ has a 
 MST of weight 0 and the vertices in 
 the remaining connection gadgets can be included
 in the MST by taking only weight-0 edges.
 This gives an MST of $G'$ of weight $k-1$.

 ($\Leftarrow$)  
 Let $S'$ be the solution for $b$-MVE on $G'$ and the MST of $G'-S'$ has weight at least $k-1$. 
 We first prove that $S' \cap Y' = \emptyset$. Let $M$ be the connection gadget between $u$ and $v$. 
 We can observe that the minimum weight of MST of $G[M \cup \{u,v\}]$ is 1. Since $M$
 is a clique of size $b+1$ and both $u$ and $v$ are connected to all vertices in $M$, the removal of arbitrary $b$ edges 
 from $G[M \cup \{u,v\}]$ cannot increase the weight of MST of $G[M \cup \{u,v\}]$.
 Hence, $S' \cap Y' = \emptyset$. Now, in order to increase the weight of MST of $G'-S'$, the interdictor must force the maximum usage of weight-1
 edges (which are available only in connection
 gadgets) in MST. To this end, we need to maximize the number of connected components in $G'-Y'$. 
 Hence, $S'$ is chosen in such a way that $G'-\{Y' \cup S'\}$ has the maximum number of connected components.
 This corresponds to {\sc Minimum $k$-Way Edge Cut} in $G$, which completes the proof.  \qed    

 \end{proof}

 \begin{theorem}\label{thm:mstfpt}
  Given an instance with edge weights only 0 and 1, $b$-{Most Vital Edges} is fixed-parameter tractable with respect to $b$.   
 \end{theorem}

 \begin{proof}
 Kawarabayashi and Thorup \cite{KT11} proved that {\sc Minimum $k$-Way Edge Cut} is FPT with respect to $s$. Here, we use their algorithm as a black box. 
 If the input graph $G$ is $d$-edge-connected with $d \leq b$, then we can find an edge cut $S$ of size at most $b$ for $G$. Since 
 $G-S$ is disconnected, we take $S$ as a solution for  {\sc $b$-MVE} and the weight of any MST of disconnected graphs is $\infty$. 
 On the other hand, if $G$ is ($b+1$)-edge-connected, we need the following claim: 

 \begin{claim}\label{lem:dconnected}
 Given a ($b+1$)-edge-connected graph, a solution of {\sc $b$-MVE} contains no weight-1 edge. 
 \end{claim}

 \begin{proof}
 Let $V_1$ and $V_2$ be an arbitrary partition of vertices of $G$ such that $V_1 \cap V_2 = \emptyset$ and $S$ be a solution of
 $b$-MVE. Let $T$ be the minimum spanning tree of $G-S$. Now, we show that if $G$ is ($b+1$)-edge-connected,~$S$ does not contain any weight-1 edge between vertices in $V_1$ and $V_2$. Since $G$ is ($b+1$)-edge-connected, there is at least one edge between $V_1$
 and $V_2$ in $G-S$. Hence, the worst-case cost of connecting $V_1$ and $V_2$ in $T$ is 1. So, it is never profitable to delete any edge of weight 1 between $V_1$ and $V_2$.   \qed 
 \end{proof}

 By this claim, if a ($b+1$)-edge-connected graph has only weight-1 edges, then the solution is empty. 
 Let $G$ be an instance of {\sc $b$-MVE}, we run the following:

\noindent
{\bfseries Step 1.} Delete all weight-1 edges from $G$. Let $G-X$ be the resulting graph where $X$ is the set of all 
      weight-1 edges in $G$. 

\noindent 
  {\bfseries Step 2.}  In each connected component of $G-X$, we run the FPT-algorithm from \cite{KT11}
      with $s$ ranging from 1 to $b$. For each connected component $C$ of $G-X$ we maintain a table A, where for each number $1 \leq i \leq b$, we
 store the maximal number of connected components that can be achieved by deleting $i$ edges in $C$. This table is of size $b \times 2$ for each connected
      component of $G-X$ and can clearly be filled in FPT time with respect to $b$.  

\noindent
 {\bfseries Step 3.} For each integer $1 \leq i \leq b$, we sort the connected components of $G-X$ according to the decreasing order of the numbers of 
       resulting components with~$i$ edge deletions as returned by Step 2. For
       each $i$ we save the top $b$ entries in this sorted list, resulting in a table $B$ of size $b \times b$. 

\noindent
{\bfseries Step 4.}  Now, we enumerate all additive partitions of $b$. The partition function~$p(b)$ gives the number of different additive partitions of $b$ without respect to 
        order which is clearly bounded by $2^{b-1}$. Such a partition can be computed in time polynomial in $b$. Let, $P_1, P_2, \cdots, P_{p(b)}$ be the additive partitions of $b$. 

\noindent
{\bfseries Step 5.} For each additive partition $P_i$ we do the following: Assume that $P_{i}$ consists of $j \leq b$ integers.
        Now for each integer $x \in P_{i}$, we branch on the first $j$ entries corresponding to $x$ from Table B, each branch assigning exactly
        one entry to $x$, that is, one connected component for $x$ from $G-X$. This will
        take $O^*( b^b)$ time for each $P_i$. Since there are at most $n$ connected components in $G-X$, Step 5 runs in $O^*(b^b)$ time.   

 \noindent{\it Correctness.}
 Correctness of Step 1 follows directly from the above claim. We prove now in Step 5 it is sufficient to branch only on the first $j$ entries corresponding 
 to $x$ from Table B. For each integer $x \in P_i$ exactly one connected component from $G-X$ is assigned. Now, the top candidate for $x$ in Table B will not be assigned to $x$
 if and only if it is assigned to another integer $y \in P_i$. There can be at most $j-1$ such integers $y \in P_i$. Hence, it is sufficient to consider only  the first $j \leq b$ 
 entries corresponding 
 to $x$ from Table B. Steps 1 and 3 can be achieved in time polynomial in $n$. Steps 2, 4 and 5 are FPT with respect to $b$. Hence we have an overall running time 
 exponentially depending on $b$.  \qed

 \end{proof}

 \section{Maximum Matching and S-T Flow Interdiction}

 In this section, we study the edge interdiction problems for maximum matching and $s$-$t$ flow problems from parameterized complexity point of view. 
 The reduction from {\sc Clique} to MMEI in \cite{abs-0804-3583} is also a parameterized one, proving that MMEI
 with respect to $b$ is W[1]-hard, even on bipartite graphs with unit edge weight and unit interdiction cost. 
 Now, we prove a similar result for the parameter $m$.

 \begin{theorem}\label{thm:mmeim}
 MMEI with unit edge weight and interdiction cost is W[1]-hard with respect to $m$. 
 \end{theorem}
 
\begin{proof}
 We give a parameterized reduction from the W[1]-hard {\sc Partial Vertex Cover} ($k$-PVC) problem with parameter $k$ \cite{GuoNW05}. Given an instance $G=(V,E)$ for $k$-PVC, we create an instance
 $G'=(V',E')$ for MMEI in the following: We initiate $G'$ with $G$ and for each vertex in $G'$, we add
 $|E|$ degree-1 neighbors. Let~$Y$ be the set of degree-1 neighbors added in this way. Next, we show that~$G$ has a set $S$
 of size $k$ which covers at least $x$ edges in~$G$ iff $G'$ has a solution~$I$ with  
 $b \leq |E|(|V|-k) + (|E|-x)$ and $\nu (G' - I)  \leq m = k$.

 ($\Rightarrow$)
 Given a solution $S$ of $k$-PVC on $G$, we construct the MMEI solution $I$ for $G'$ as follows: We 
 add all edges in $G'$ which are not incident to any vertex in $S$ to $I$. 
 Since $S$ covers at least $x$ edges in $G$, we add at most $|E|-x$ edges from $E$ and $|E|(|V|-k)$ edges between $Y$ and $V$ to $I$.
 In the subgraph $G'-I$,  every edge is incident to vertices in $S$.
 Hence
 $\nu(G'-I)$ is at most $k$.

 ($\Leftarrow$)
 Let $I$ be the given solution of MMEI for $G'$ with $b \leq |E|(|V|-m) + (|E|-x)$ and $\nu (G' - I) \leq m$. 
 Since $\nu (G' - I) $ is at most $m$, at most $m$ vertices in $G'[V' \setminus Y]-I$ can have 
 degree-1 vertices
 attached to them. Let $X$ denote the set of vertices in $G'[V' \setminus Y]-I$ which have degree-1 neighbors.
 To remove all degree-1 neighbors of the vertices in $V' \setminus X$ requires addition of $|E|(|V|-m)$ edges to $I$. Hence, the vertices in $X$
 must cover at least $x$ edges in $G$ and $X$ is solution for $k$-PVC for $G$.     \qed

 \end{proof}

However, unlike for the parameter $b$, the parameterization of MMEI with $m$ becomes tractable, if we restrict the input graphs to be bipartite.

 \begin{theorem}\label{thm:mmeibpg}
 MMEI is FPT with respect to~$m$ when restricted to bipartite graphs with unit edge weight and interdiction cost. 
 \end{theorem}
 
 \begin{proof}
 
 We prove the theorem by showing that, for a bipartite graph $G=(X,Y,E)$, there is a partial vertex cover $S$ with $|S| \leq k$ covering at least $x$ edges, if
and only if there is a set $I \subseteq E$ with $|I| \leq |E|-x$ and $\nu (G-I) \leq m= k$. Note that $k$-PVC on bipartite graphs is solvable in 
 $O^*(2k^{(2k)})$ time \cite{AminiFS11}, proving the theorem. 

Let $S$ be a size-$k$ partial vertex cover of $G$ and $I$ be the set of edges not incident to the vertices in $S$. Then, $|I| \leq |E|-x$. Since,
all edges in $G-I$ are incident to vertices in $S$ and each vertex in $S$ must have an incident edge whose other endpoint is not in $S$, $\nu (G-I) \leq k$. This is
true because, if  each vertex in $S$ does not have an incident edge whose other endpoint is not in $S$, then the cover can be smaller. The reverse 
direction can be shown in similar way. \qed

 \end{proof}
Using both $b$ and $m$ as parameters, we can achieve another FPT result for MMEI.

\begin{theorem}\label{thm:mmfpt}
MMEI parameterized by both $b$ and $m$ is FPT in graphs with unit edge weight.

\end{theorem}
 
 \begin{proof}
 We show that in the instances with unit weight MMEI with both $b$ and $m$ as parameters admits a kernel. We apply the following reduction rules:\\

\noindent
{\bfseries Reduction Rule 1:}
 If a vertex $v$ has more than $b$ degree-1 neighbors, then keep $b$ of them and remove other degree-1 neighbors of $v$.

The correctness of Rule 1 can be shown as follows:
Assume that a vertex $v$ has more than $b$ degree-1 neighbors. Let $X$ be the set of all degree-1 neighbors of $v$. It is not possible to remove all edges between
 $v$ and $X$ with at most $b$ edge deletions allowed. Then, one
 of the edges between $v$ and $X$ can be in the matching. Now, keeping $b$ of them obviously does not 
 omit any optimal solution. \\

\noindent
{\bfseries Reduction Rule 2:} If $U \subseteq V$ and $W \subseteq V$ with $U \cap W \neq \emptyset$ satisfy: 
1) $W$ is an independent set, 
2) $N(W) = U$ and there are all possible edges between $U$ and~$W$, and
3) $|W| \geq \text{max}\{|U|, b+1\}$,
then keep only max$\{|U|, b+1\}$ vertices in $W$ and remove the rest.

We prove now the correctness of Rule 2.
Let $H$ be the bipartite graph $G[U \cup W]-E(G[U])$. Notice that we have $\nu(H)=|U|$ and there are at least $|W|$ disjoint matchings in $H$ which saturate $U$. Now we show that removing any $b$
edges from $H$ does not decrease the cardinality of a maximum matching in $G$. This property is obtained by observing that since in $H$ there are at least $b+1$ disjoint matchings which saturate $U$, 
we have after removing up to $b$ edges in~$G$ there is at least one matching of $H$ which saturate $U$. We therefore have the desired property that $H$ is immune to ``edge removals''. Hence, removing
all but $\text{max}\{b+1,|U|\}$ vertices from $|U|$ still maintains this property.

Rule 2 runs in polynomial time, since $W$ is clearly a module and all modules of a graph can be found in linear time \cite{HsuM91}.

 \begin{claim}
 MMEI with both $b$ and $m$ as parameters admits a kernel.
 \end{claim}

 \begin{proof}
 Let $B$ be the set of edges which form the solution of MMEI and let~$M$ be the maximum matching 
 of the remaining graph. Since the cost of each edge are positive integers, $|B| \leq b$ and $|M| = m$.
 Moreover, $|V(B)| \leq 2b$ and~$|V(M)|~=~2m$. 
 As $M$ is the maximum matching in $G-B$, each edge in $G-B$ must have at least one endpoint in $V(M)$. Hence, there exist at most $b$ edges in $G$ which 
 do not have its endpoints in $V(M)$.

 Now, we bound the number of edges with its endpoints in $V(M)$. 
 There are at most $4m^2$ edges in $G[V(M)]$. Let $X$ be the set of degree-1 neighbors of~$V(M)$. Rule 1 bounds the 
 number of degree-1 neighbors of each vertex by~$b$; hence, there are at most $2bm$ edges between $V(M)$ and $X$. The number of edges between
 $V(B)$ and $V(M)$, is clearly bounded by $4bm$. Now, the remaining edges are between~$V(M)$ and $Y:=V \setminus (X \cup V(B) \cup V(M)$.
 We can observe that the vertices in $Y$ have degree at least two and $N(Y) \subseteq V(M)$. 
 For each edge~$\{m_1, m_2\} \in M$, there exists at most one vertex in $Y$ which is adjacent to both $m_1$ and $m_2$. We have $m$ edges in $M$, hence
 there are at most $m$ vertices in~$Y$ which are adjacent to both endpoints of an edge in $M$. Next, we bound the number of vertices in~$Y$, which are adjacent
 to several edge in $M$. By Rule~2, for a size-$i$ subset of vertices $I \subset V(M)$ such that no two vertices in $I$ are connected by an edge in $M$, there can be at most max$\{i, b+1\}$ vertices $y$ 
 in $Y$, such that~$N(y)=I$. There can be at most $ \sum_i 2^i  \cdot {m \choose i}$ such subsets $I$ in $V(M)$. Hence we have at most
 $l= \sum_i 2^i \cdot {m \choose i} \cdot \text{max} \{i, b+1\}~+~m$ vertices in $Y$. Hence, there are at most $lm$ edges between $Y$ and $M$. In total
 we have $lm + 4bm + 4m^2$ edges which is a function depending only on $m$ and $b$. Hence we have a kernel for MMEI with both $b$ and $m$ as parameters.  \qed 
 \end{proof}

 \end{proof}

Wood \cite{Wood93} proved the NP-hardness of $s$-$t$ FEI by a reduction from {\sc Clique}, which sets the flow amount in the resulting graph $r$ equal to $k$. This implies that~$s$-$t$ Flow Edge Interdiction 
with unit edge cost and the edge capacity being~1 or 2 is W[1]-hard with respect to $r$.
Complementing this result, we achieve the W-hardness of $s$-$t$ FEI for parameter $b$.

 \begin{theorem}\label{thm:stwhardb}
  $s$-$t$ {\sc Flow Edge Interdiction} is W[1]-hard with respect to $b$ even with bipartite graphs with unit flow capacity as input.

 \end{theorem}

 \begin{proof}
 We give a parameterized reduction from the W[1]-hard MMEI with $b$ as parameter. Note that this parametrization remains W[1]-hard
on bipartite graphs. 
 Let a bipartite graph $G=(X,Y,E)$ be an instance of MMEI. We create an instance $G'$ for $s$-$t$ {\sc Flow Edge Interdiction}
 in the following way: Initialize~$G'$ with $G$ such that each edge has unit interdiction cost
 and flow capacity and each arc is directed from vertex in~$X$ to the one in $Y$. Add 
 two new vertices~$s$ and $t$ to $G$. 
 Now, we arcs with interdiction cost $b+1$ from~$s$ to each vertex in~$X$; similarly, arcs with interdiction cost $b+1$ directed from 
 each vertex in $Y$ to $t$ are added. Let the set of the 
 arcs added in this way to $G$ be $Q$ and each arc in $Q$ has unit flow capacity.
 With this construction we can show that~$G$ has a yes answer to MMEI with $b$ total budget and maximum matching with weight at most $m$
 allowed in the resulting graph  iff $G'$ has yes answer to the $s$-$t$ {\sc FEI} with $b$ total budget and maximum flow allowed in the resulting graph
 at most $r=m$.

 The key argument is that only the arcs in $G'-Q$ will belong to an optimal solution of
 $s$-$t$ {\sc FEI}. 
  Hence, the amount of the $s$-$t$ flow in the resulting graph is equivalent to the corresponding matching in $G'-Q$.  \qed
 \end{proof}

\section{Partial Problems on Bipartite Graphs}
From the proof of Theorem \ref{thm:mmeibpg}, we can already observe some equivalent relation between edge interdiction problems and partial covering problems. 
In the following, we introduce a new edge interdiction problem and prove
its parameterized complexity by relating it to a partial covering problem.

 \begin{description}\itemsep-1pt
  \item {\sc Minimum Maximal Matching Edge Interdiction (MMMEI)}  
  \item {Input:} A simple graph $G=(V,E)$, and an integer interdiction budget $b \geq 0$ and an integer $r$.  
  \item {  Output:} Is there a  subset $I \subseteq E$ with $|I| \leq b$ such that $\lambda(G - I) \leq r$ ?  
 \end{description}

The corresponding partial covering problem is the so called {\sc Partial Edge Dominating set} ($k$-PEDS)
problem which is defined as follows: 

 \begin{description}\itemsep-1pt
  \item {\sc $k$-Partial Edge Dominating Set ($k$-PEDS)}  
  \item {Input:} A graph $G=(V,E)$ and two positive integers $k$ and $x$.  
  \item {  Output:} Is there a  subset $S \subseteq E$ with $|S| \leq k$ such that at least $x$ edges are dominated by $S$?  
 \end{description}

\begin{lemma}\label{peds} 
 $k$-PEDS with parameter $k$ is W[1]-hard on bipartite graphs.
\end{lemma}

\begin{proof} 
We give a parameterized reduction from W[1]-hard $k$-{\sc Independent Set}~\cite{Nie06} to $k$-{\sc PEDS}.
Given a graph $G=(V,E)$ as an instance of $k$-{\sc Independent Set}, we create an instance 
$G':=(V_1,V_2,E')$ for $k$-PEDS in the following way: For each vertex $v \in G$, we create
two vertices $v_1$ and $v_2$ and an edge $\{v_1, v_2\}$ in~$G'$. For each edge
$\{u,v\} \in G$, we create two edges $\{u_1, v_2\}$ and $\{u_2, v_1\}$ in $G'$. 
Moreover, for every vertex $v_i \in G'$, we add $n - deg_{G}(v)$ degree-1 neighbors where~$n=|V|$. Now, we show that $k$ edges in $G'$ dominate at least $2kn$ edges 
iff there exists an independent set of size $k$ in $G$.

($\Leftarrow$)
Let $S$ be an independent set of size $k$ in $G$. 
For each vertex $v \in S$, we add the corresponding edge 
$\{v_1, v_2\} \in E'$ to the solution set $S'$ for $k$-PEDS in $G'$.
Given that $S$ is an independent set, for any pair of vertices $u$, $v$ in $S$, the corresponding edges 
$\{u_1, u_2\} \in E'$ and $\{v_1, v_2\} \in E'$ do not dominate any common edge.
Hence, since each edge $e \in S'$ dominates exactly $2n$ edges, the set $S'$ dominates $2kn$ edges.    

($\Rightarrow$)
Now, let $S'$ be a set of $k$ edges in $G'$ which dominate $2kn$ edges. Each edge in $G'$ can dominate at most $2n$ 
edges; hence, no two edges in~$S'$ share a dominated edge. This ensures that the shortest distance
between every two edges in~$S'$ is at least two. Now, we present an algorithm to convert a given solution~$S'$ for PEDS
in $G'$ to a size-$k$ independent set in $G$. For this purpose, first we define a {\it conflict cycle}. 
Let~$S'$ be a solution of PEDS in $G'=(V_1,V_2,E')$ and let~$T \subseteq S'$. We say
that $T$ forms a conflict cycle $C$ in $G'$ if we can construct a cycle in $G'$ containing $T$ and a set $U$ of $|T|$ 
vertices from
$G' - S'$ such that in cycle $C$, between any two edges from $T$ there exists exactly one vertex from $U$. A vertex~$t_j$ can be in $U$ only if~$t_{3-j}$ is contained
 in~$S'$ for $j=\{1,2\}$.

Assume that there exists no conflict cycle with respect to $S'$. 
Then it is easy to get the size-$k$ independent set corresponding to $S'$ in $G$. Construct a graph~$Y$ that represents 
the connectivity relation of edges in $S'$: For each edge $i \in S'$, we create a  vertex $y_i$ in $Y$.
We create an edge between two 
vertices $y_i$ and $y_j$ in $Y$ iff their corresponding edges $i$ and $j$ are separated
by distance exactly 2 in $G'$. Observe that in the absence of {\it conflict cycles} in $G'$, $Y$ is a tree.
Now, we give a procedure to get $S$ from $S'$, given $Y$ is a tree. We start in bottom-up fashion from leaves.
Consider a leaf of $Y$, if $\{x_1,x_2\}$ is the edge corresponding to the leaf, then we add the corresponding vertex $x$ to $S$. 
Let $\{x_1, y_2\}$ be the edge corresponding to the leaf and $v_i$ be the vertex connecting  the leaf to its parent in $T$. 
If $v_i \in V_1$, we add $x$ to $S$, else we add $y$ to $S$. Now, we remove this leave and proceed iteratively for every leaf. We
can observe that, since $T$ is a tree, no conflict will arise during this procedure. The
obtained solution after all vertices in $T$ are processed is an independent set in $G$. 
In the scenario when there exist {\it conflict cycles} in $G'$, we first prove the following claim:

\begin{claim}
If there exist {\it conflict cycles} in $G'$, the corresponding graph $Y$ is a bipartite graph.
\end{claim}
\begin{proof}
Firstly, we observe that $T \subseteq S'$ forms a {\it conflict cycle}, only if there exists no edge $\{x_1, x_2\}$ in $T$. If $\{x_1, x_2\} \in T$, 
no vertex $v_i$ with $i\in \{1,2\} $ adjacent to~$x_1$ or $x_2$ can be in $U$ as $v_{(3-i)} \notin V(S') $. We can further observe that
since $G'$ is bipartite, in the {\it conflict cycle}, the vertices will alternate between $V_1$ and $V_2$. Moreover, in the cycle, every two consecutive
vertices from $U$ will also alternate between $V_1$ and $V_2$. Now, as the edges in $Y$ are analogous to vertices in $U$, each cycle in $Y$ corresponding
to a {\it conflict cycle} in $S'$ is of even length.       \qed
\end{proof}
 
For all vertices in $T$ which do not belong to any cycle, we can obtain the corresponding vertices in $S$
in the bottom up fashion recursively as for the case that $T$ is tree. Now after all vertices not belonging to any cycles are dealt with, let the
resulting $T$ be $T'$, where only even cycles remain.
 Since $T'$ is bipartite, it is 2-colorable. We color $T'$ with two colors, say black and white. If a vertex in~$T'$ corresponding to $\{x_1, y_2\}$
is black, we add $x$ to $S$, else $y$. Since the vertices in $U$ alternate between $V_1$ and $V_2$, this resolves all conflicts and gives a valid independent set $S$.  \qed

\end{proof}

\begin{theorem}
{\sc Minimum Maximal Matching Edge Interdiction} is W[1]-hard with parameter $b$ even on bipartite graphs. 
\end{theorem}
\begin{proof}

For a bipartite graph $G=(X,Y,E)$, we show that there is a set $S \subseteq E$ with $|S|=k$ and dominating $x$ edges, iff there is a set $I \subseteq E$ with $|I|=|E|-x$ and $\lambda (G-I) = m =k$.
The theorem then follows from Lemma \ref{peds}.
Let $I$ be the set of edges not dominated by $S$, $|I|=|E|-x$. Clearly, $S$ is the minimum edge dominating set of $G - I$. It is well-known that the size of the minimum edge 
dominating set of a graph is equal to the size of its
minimum independent edge dominating set. In fact, given a minimum edge dominating set $F$ of $G-I$ we can construct in polynomial time a minimum independent edge dominating set of $G-I$
 \cite{Yannaeds}. Moreover, a minimum independent edge dominating set is also a minimum maximal matching of $G-I$. Hence, $G-I$ has a minimum maximal matching of size $m$. 
The reverse direction can be shown similarly. \qed
\end{proof}

We can also use the equivalence relation between edge interdiction problems and partial covering problems to prove hardness results of partial covering problems.
 
\begin{corollary}\label{coro:pvc} 
 [*]\footnote{ Due to space limitations, proofs of results marked [*] are given in Appendix} $k$-PVC on bipartite graphs is W[1]-hard with respect to the number of uncovered edges.
\end{corollary}

Finally, we study another bipartite variant of the {\sc Partial vertex Cover} problem, which could be of independent interest. This variant is called ~{$(k_1,k_2)$-{\sc Partial Vertex Cover} ($(k_1,k_2)$-PVC)},
 where given   
a bipartite graph~$G=(X,Y,E)$ and~$k_1,~k_2~\in~\mathbb{N}$, one asks for  
a subset~$S \subseteq V$ with~$|S \cap X|=k_1$ and~$ |S \cap Y|= k_2$ that maximizes the number of edges in~$G$ with at 
                    least one endpoint in~$S$.  
In contrast to the fixed parameter tractability of~$k$-PVC on bipartite graphs, we prove that~$(k_1,k_2)$-PVC is W[1]-hard with respect
to~$k_1$ and~$k_2$. To this end, we prove first a so-called sparsest subgraph problem which is W-hard and reduce it to~$(k_1,k_2)$-PVC. This problem is
called~$(k_1,k_2)$-{\sc Sparsest Subgraph}~($(k_1,k_2)$-SS), where given a bipartite graph $G=(X,Y,E)$ and  $k_1$, $k_2 \in \mathbb{N}$, one
asks for a subset $S \subseteq V$ with $|S \cap X|=k_1$ and $ |S \cap Y|= k_2$ that minimizes the number of edges in $G[S]$.

\begin{lemma}\label{thm:kss}
[*] $(k_1, k_2)$-SS is W[1]-hard with respect to $k_1$ and $k_2$.
\end{lemma}

 \begin{theorem}\label{thm:kkpvc}
[*] $(k_1, k_2)$-PVC is W[1]-hard with respect to $k_1$ and $k_2$. 
 \end{theorem}

\section{Outlook}
We proved that $b$-MVE is FPT with $b$  as parameter for the case of edge weights 0 or 1. The case with integer positive weights remains open. Another open
question is the complexity of MMEI with $b$ and $m$ as parameters  and integer edge weights. Moreover, structural parameters like treewidth could be a promising alternative for parameterizing
interdiction problems. Finally, the vertex interdiction problems  have been studied from the viewpoints of classical complexity and approximation algorithms, but seem unexplored from the parameterized complexity perspective.
\bibliographystyle{abbrv}	 
\bibliography{interdict}

\begin{thebibliography}{10}

\bibitem{AminiFS11}
O.~Amini, F.~V. Fomin, and S.~Saurabh.
\newblock Implicit branching and parameterized partial cover problems.
\newblock {\em J. Comput. Syst. Sci.}, 77(6):1159--1171, 2011.

\bibitem{AroraK06}
S.~Arora and G.~Karakostas.
\newblock A 2 + $\epsilon$ approximation algorithm for the $k$-mst problem.
\newblock {\em Math. Program.}, 107(3):491--504, 2006.

\bibitem{n87}
N.~Assimakopoulos.
\newblock A network interdiction model for hospital infection control.
\newblock {\em Bio. Med}, 17(6):413--422, 1987.

\bibitem{Bar-Yehuda01}
R.~Bar-Yehuda.
\newblock Using homogeneous weights for approximating the partial cover
  problem.
\newblock {\em J. Algorithms}, 39(2):137--144, 2001.

\bibitem{gupta}
M.~Dinitz and A.~Gupta.
\newblock Packing interdiction and partial covering problems.
\newblock In {\em IPCO}, pages 157--168, 2013.

\bibitem{DowneyEFPR03}
R.~G. Downey, V.~Estivill-Castro, M.~R. Fellows, E.~Prieto, and F.~A. Rosamond.
\newblock Cutting up is hard to do: the parameterized complexity of k-cut and
  related problems.
\newblock {\em Electr. Notes Theor. Comput. Sci.}, 78:209--222, 2003.

\bibitem{FominLRS11}
F.~V. Fomin, D.~Lokshtanov, V.~Raman, and S.~Saurabh.
\newblock Subexponential algorithms for partial cover problems.
\newblock {\em Inf. Process. Lett.}, 111(16):814--818, 2011.

\bibitem{forcade}
R.~Forcade.
\newblock Smallest maximal matching in the graph of the d-dimensional cube.
\newblock {\em J. Combinatorial Theory Ser. B}, 14(14):153--156, 1973.

\bibitem{FredericksonS99}
G.~N. Frederickson and R.~Solis-Oba.
\newblock Increasing the weight of minimum spanning trees.
\newblock {\em J. Algorithms}, 33(2):244--266, 1999.

\bibitem{GuoNW05}
J.~Guo, R.~Niedermeier, and S.~Wernicke.
\newblock Parameterized complexity of generalized vertex cover problems.
\newblock In {\em WADS}, pages 36--48, 2005.

\bibitem{HsuM91}
W.-L. Hsu and T.-H. Ma.
\newblock Substitution decomposition on chordal graphs and applications.
\newblock In {\em ISA}, pages 52--60, 1991.

\bibitem{KT11}
K.~Kawarabayashi and M.~Thorup.
\newblock The minimum k-way cut of bounded size is fixed-parameter tractable.
\newblock In {\em FOCS}, pages 160--169, 2011.

\bibitem{Liang01}
W.~Liang.
\newblock Finding the k most vital edges with respect to minimum spanning trees
  for fixed k.
\newblock {\em Discrete Applied Mathematics}, 113(2-3):319--327, 2001.

\bibitem{LiangS97}
W.~Liang and X.~Shen.
\newblock Finding the k most vital edges in the minimum spanning tree problem.
\newblock {\em Parallel Computing}, 23(13):1889--1907, 1997.

\bibitem{Nie06}
R.~Niedermeier.
\newblock {\em Invitation to Fixed Parameter Algorithms (Oxford Lecture Series
  in Mathematics and Its Applications)}.
\newblock {Oxford University Press, USA}, March 2006.

\bibitem{pan}
F.~Pan and A.~Schild.
\newblock Interdiction problems on planar graphs.
\newblock In {\em APPROX-RANDOM}, pages 317--331, 2013.

\bibitem{philip}
C.~A. Phillips.
\newblock The network inhibition problem.
\newblock In {\em Proceedings of the twenty-fifth annual ACM symposium on
  Theory of computing}, STOC '93, pages 776--785, New York, NY, USA, 1993. ACM.

\bibitem{salmeron_09}
J.~Salmeron, K.~Wood, and R.~Baldick.
\newblock Worst-case interdiction analysis of large-scale electric power grids.
\newblock {\em IEEE Transactions on Power Systems}, 24(1):96--104, Feb. 2009.

\bibitem{Wood95}
A.~Washburn and R.~K. Wood.
\newblock Two-person zero-sum games for network interdiction.
\newblock 43(2):243--251, 1995.

\bibitem{ghare71}
R.~K. Wood.
\newblock Optimal interdiction policy for a flow network.
\newblock 18:37--45, 1971.

\bibitem{Wood93}
R.~K. Wood.
\newblock Deterministic network interdiction.
\newblock 17(2):1--18, 1993.

\bibitem{Yannaeds}
M.~Yannakakis and F.~Gavril.
\newblock Edge dominating sets in graphs.
\newblock {\em SIAM J. Appl. Math}, 38(3):364--372, 1980.

\bibitem{abs-0804-3583}
R.~Zenklusen.
\newblock Matching interdiction.
\newblock {\em CoRR}, abs/0804.3583, 2008.

\bibitem{Zenklusen10a}
R.~Zenklusen.
\newblock Matching interdiction.
\newblock {\em Discrete Applied Mathematics}, 158(15):1676--1690, 2010.

\end{thebibliography}
\newpage
\appendix
\section{Appendix}

\subsection{Omitted proof of Corollary \ref{coro:pvc}}
\begin{proof}
Again, we prove this theorem by establishing an equivalence relation between MMEI and PVC. For a bipartite graph $G=(X,Y,E)$, we show that there is a set $S \subseteq V$ with $|S| \leq k$
and covering $|E|-x$ edges iff there is a set $I \subseteq E$ with $|I|= x$ and $m = \nu(G-I)=k$. The theorem follows from W-hardness of MMEI with unit interdiction cost and weight on bipartite graphs
with parameter~$b$~\cite{abs-0804-3583}.

Let $I$ be the set of uncovered edges of $S$.
 We can observe that since $G-I$ have a minimum vertex cover of $k$, from K\"{o}nig's theorem, $G-I$ has a maximum matching of size at most $k$. 
The reverse direction is obvious.    \qed

\end{proof}

\subsection{Omitted proof for Lemma \ref{thm:kss}}

\begin{proof}
We give a parameterized reduction from W[1]-hard {\sc $k$-Clique} \cite{Nie06}. Given an instance $G=(V,E)$ of $k$-{\sc Clique},
we create an instance~$G'=(X,Y,E')$ of~$(k_1,k_2)$-SS as follows: For each 
vertex~$v_i \in V$, we create a vertex~$v_i$ in~$X$ and for each edge $\{v_i,v_j\} \in E$, we create a vertex $v_{ij} \in Y$. For each
vertex~$v_k \in X$, we create an edge $\{v_k, v_{ij}\} \in E'$ iff $k \neq \{i,j\}$. 
Now, we show that $G$ has a~$k$-{\sc Clique} iff there exists a $(k_1,k_2)$-SS in $G'$ which induces at most $k_1k_2-2k_2$
edges where $k_1=k$ and $k_2= {k \choose 2}$.

$(\Rightarrow)$ 
Let $S \subseteq V$ be the size-$k$ clique in $G$. Now, for each vertex $v_i \in S$, we add the 
corresponding vertex $v_i \in X $ and the vertices in $Y$ corresponding to the edges in $G[S]$ to $S'$.    
We can observe that exactly $k_1=k$ vertices are selected from $X$ and $k_2= {k \choose 2}$ vertices are selected from $Y$. Moreover, since $S$ is a clique,~$G'[S']$ is
a $(k_1, k_2)$-biclique in $G'$ minus the edges corresponding to vertex-edge connections in $G[S]$. 
There are exactly $2k_2$ missing edges since $G[S]$ is a size-$k$ clique. Hence, $G'[S']$ has $k_1k_2-2k_2$ edges.

$(\Leftarrow)$
Let $S' \subseteq V'$ be a $\{k_1,k_2\}$-SS for $G'$. Moreover, let $S' \cap X= S'_{1}$ and~$S' \cap Y=S'_{2}$. Since there are at most $k_1k_2-2k_2$ edges in $G'[S']$, 
$G'[S']$ is a $(k_1, k_2)$-biclique in $G'$ with $2k_2$ edges missing. 
This means that $k_2={k \choose 2}$ edges are between $k_1=k$ vertices in $G$, which is possible
iff there exists a size-$k$ clique in $G$.    \qed   

\end{proof}

 \subsection{Omitted proof of Theorem \ref{thm:kkpvc}}
\begin{proof}
 In Lemma \ref{thm:kss}, we proved that~$(k_1,k_2)$-SS is W[1]-hard with respect to~$k_1$ and~$k_2$. 
 Now, we give parameterized reduction from 
 this problem to~$(k_1,k_2)$-PVC. 
 Given an instance~$G=(X,Y,E)$ of~$(k_1,k_2)$-SS with vertex set of size~$n$, we create an instance~$G'=(X',Y', E')$
 of~$(k_1,k_2)$-PVC from~$G$ by adding~$n-\text{deg}_{G}(v)$ degree-1 vertices
 to each vertex of~$G$. Let~$L$ be the set of all {degree-1} vertices added. We need to show that a~$(k_1,k_2)$-SS
 for~$G$ which induces minimum number of edges, is equivalent to~$(k_1,k_2)$-PVC for~$G'$ 
 which covers the maximum number of edges in~$G'$. We can assume that no degree-1 vertex is part of~$(k_1,k_2)$-PVC 
 solution as it is always better to take its neighbor into the solution. Each vertex~$v$ in~$X' \cup Y' \setminus L$ has
 the same degree~$n$. Hence, it is easy to see that a $(k_1,k_2)$-PVC set induces the minimum number of edges as the 
 induced edges are shared 
 edges between two vertices of the PVC-set. With this argument we can show that $(k_1,k_2)$-SS
 for~$G$ is equivalent to $(k_1,k_2)$-PVC for $G'$ .   \qed 
\end{proof}

\end{document}